\documentclass[%
 aip,
 jmp, 
 amsmath,amssymb,
 reprint,  
 onecolumn, 
 hyperref,
]{revtex4-1}

\usepackage{graphicx}
\usepackage{dcolumn}
\usepackage{bm}

\usepackage{graphicx,color}
\usepackage{epstopdf}
\usepackage{amsmath,amsthm,amsfonts,amssymb,times}
\usepackage{mathrsfs}

\def\be{\begin{equation}}     
\def\ee{\end{equation}}

\def\ket#1 {| #1 \rangle}

\numberwithin{equation}{section}

\newtheorem{theorem}{Theorem}
\newtheorem{corollary}[theorem]{Corollary}

\newtheorem{proposition}[theorem]{Proposition}

\theoremstyle{definition}



\begin{document}
\title{Exponential vanishing of the ground-state gap of the QREM  via adiabatic quantum computing}
\author{J. Adame}
 \affiliation{Department of Mathematics, Caltech, Pasadena, USA.}
\author{S. Warzel}%
 \email{warzel@ma.tum.de.}
\affiliation{ 
Zentrum Mathematik, TU M\"unchen, Boltzmannstr. 3, 85747 Garching, Germany.
}%
\date{\today}

\begin{abstract} 
In this note we use  ideas of Farhi, Goldstone, Gosset, Gutmann, Nagaj and Shor who link a lower bound on the run time of their quantum adiabatic search algorithm to an upper bound on the energy gap above the ground-state of the generators of this algorithm. We  apply these ideas to the quantum random energy model (QREM). Our main result is a simple proof of the conjectured exponential vanishing of the energy gap  of the QREM.
\end{abstract}

\maketitle

\section{\label{sec:level1} Quantum Search Algorithms}

Finding the minimum value in an unstructured energy landscape $ u: \{ 1, \dots, M\} \to \mathbb{R} $ is a task which  by any classical algorithm generally amounts to order $ M  $ trials  to succeed. Ever since Grover proposed his algorithm, it is   known that 
this search can be sped up by a factor of $ \sqrt{M} $ through quantum computations.\cite{Gr1,Gr2} Shortly after,
Farhi and collaborators\cite{FGGS00,F+01} proposed another quantum search algorithm which has the advantage of being based on the continuous time-evolution  without using quantum gates.
Their idea was to encode the energy landscape $ u $ 
in a diagonal matrix $ U = {\rm diag}\left(u(1),\dots,u(M) \right) $, 
which is sometimes referred to as the `Problem-Hamiltonian',  and acts on $ \mathbb{C}^M $.
The task of finding a minimum 
is now equivalent to the search for a ground-state  of $ U $. 
To accomplish this, the authors suggested to use the quantum evolution with an adiabatic time-scale $ T > 0 $ 
\begin{equation}\label{eq:qdyn}
i \frac{d}{dt} \psi(t) = h(t/T) \psi(t) \, , \qquad \psi(0)  \in \mathbb{C}^M \, .
\end{equation}
The time-dependent generators are taken to be of the form
$$
h(s) = h_D(s) + c(s) \, U \, ,
$$
satisfying the following assumptions:
\begin{itemize}
\item[\bf a1] $ c:  [0,1] \to [0,1] $ is twice-continuously differentiable with $ c(0) = 0 $ and $ c(1) = 1 $,
\item[\bf a2] $ h_D: [0,1]\to {\rm Herm}(\mathbb{C}^{M\times M}) $  is twice-continuously differentiable with $ h_D(1) = 0 $,
\item[\bf a3]  $ h(s) $ has a non-degenerate ground-state $ \phi(s) \in \mathbb{C}^M $ for any $ s \in [0,1] $.
\end{itemize}
Since $ h(1) = U $, this in particular requires $ u $ to have a unique minimum which we will denote by $ u(j_0) = \min_k u(k) $. 
The quantum search for this minimum then amounts to starting the time-evolution~\eqref{eq:qdyn} in the known ground-state $ \phi(0) = \psi(0) $  of the 'Driving-Hamiltonian' $ h_D(0) = h(0) $, and reading out the components  of the state  $ \psi(T) $ at the final time  in the canonical basis $ e_1, \dots e_M \in \mathbb{C}^M $. If the adiabatic time $ T > 0 $ is large enough, the hope is to arrive  in the unique ground-state $ \phi(1) = e_{j_0} $ of $ h(1) = U $. 
More quantitatively, the probability $ \left| \langle \phi(1) ,  \psi(T) \rangle \right|^2 =  \left| \langle e_{j_0} ,  \psi(T) \rangle \right|^2 $ that  the time-evolution ends up in the state $ \phi(1) = e_{j_0} $ is estimated with the help of the adiabatic theorem of Kato.\cite{Kato} The following is an explicit version taken from Ref.~\onlinecite{JRS07}.
\begin{theorem}[cf.~Ref.~\onlinecite{JRS07}]\label{thm:adabatic}
Let $ h: [0,1]\to {\rm Herm}(\mathbb{C}^{M\times M}) $ be a family of twice continuous differentiable hermitian matrices with
\begin{enumerate}
\item a non-degenerate ground-state  $ \phi(s) \in  \mathbb{C}^M $, and 
\item an energy-gap 
$ \gamma(s) > 0 $ above the ground-state. 
\end{enumerate} 
Then the unique solution of the initial-value problem~\eqref{eq:qdyn}
satisfies:
\begin{align}
 \sqrt{1 - \left| \langle \psi(T) , \phi(1) \rangle \right|^2 }
  \leq  \frac{1}{T} \left[ \frac{1}{\gamma(0)^2} \| h'(0) \|  + \frac{1}{\gamma(1)^2} \| h'(1) \|   + \int_0^1 \frac{7}{\gamma(s)^3} \| h'(s) \|^2 + \frac{1}{\gamma(s)^2} \| h''(s) \|  ds \right]\, , 
\end{align}
where $ \| \cdot \| $ denotes the operator norm. 
\end{theorem}

The adiabatic theorem hence ties the value of the energy gap $ \gamma(s) $ to the run time of the quantum adiabatic search. Usually, bounds on the energy gap are used to estimate the run time. 
 In this paper, we follow an idea of Fahri, Goldstone, Gosset, Gutmann, and Shor\cite{FGGGS10} to deduce a bound on the smallest gap, $ \min_{s \in [0,1]} \gamma(s) $, from a lower bound on the run-time of the quantum adiabatic algorithm. Our main novel point presented in Section~\ref{sec:QREM} below is the application of this idea to the quantum random energy model (QREM). Before, presenting the QREM let us summarize the results of Refs.~\onlinecite{FGGN08,FGGGS10} needed below.

\subsection{\label{sec:level2} The scrambled ensemble and a lower bound on the run time}

Initially, the aim was to outperform the Grover algorithm in the set-up of quantum adiabatic computing described above. In particular, in case of search problems which belong to the NP-complete class the hope was to have identified a quantum search algorithm which has polynomial run time. That this is not the case was realised shortly after. From a computational complexity point of view the above quantum search algorithm  is equivalent to all other models for universal quantum computation.\cite{AD+04} 

Farhi, Goldstone, Gutmann, and Nagaj\cite{FGGN08} later quantified this fact through the following lower bound on the run time of the algorithm. 
More specifically, their result concerns scrambled versions of the original search problem,
\begin{equation}
U_\pi = {\rm diag}\left(u(\pi^{-1}(1)),\dots,u(\pi^{-1}(M)) \right) \, ,  
\end{equation}
where $ \pi \in \mathcal{S}_M $ denotes a permutation of the $ M $ elements. Let $ h_\pi(s) :=  h_D(s) + c(s) \, U_\pi $ denote the generator of the scrambled-time adiabatic evolution 
\begin{equation}\label{eq:qdyn2}
 i \frac{d}{dt} \psi_\pi(t) = h_\pi(t/T) \psi_\pi(t)  
 \end{equation}
 each starting from the same initial state $ \psi_\pi(0) = \phi(0) \in \mathbb{C}^M $. The set of permutations for which the quantum adiabatic search succeeds with probability $ b  \in [0,1] $ will be denoted by
 \begin{equation}
 \mathcal{P}_{M}(b) = \left\{ \pi \in \mathcal{S}_M \, | \, \left| \langle e_{\pi(j_0)} , \psi_\pi(T) \rangle \right|^2 \geq b \right\} \, . 
 \end{equation}
 Note that the minimum value corresponding to $ U_\pi $ is now found in the $ \pi(j_0) $th entry on the diagonal.
Clearly, the number $ | \mathcal{P}_{M}(b)  | $ of such permutations is less or equal to the total number of permutations $ M! $. Knowing that $ | \mathcal{P}_{M}(b)  | $ consitutes a substantial fraction is enough to deduce a lower bound on the run-time $ T > 0 $ of the quantum adiabatic search. 
\begin{theorem}[cf.~Ref.~\onlinecite{FGGN08}]\label{thm:QAC}
Consider the scrambled quantum-time evolution~\eqref{eq:qdyn2} with generators $  h_\pi(s) =  h_D(s) + c(s) \, U_\pi $ satisfying {\bf a1-2} and common initial state $ \psi_\pi(0) = \phi(0) \in \mathbb{C}^M $.  
Suppose that for some $ \varepsilon , b \in (0,1)$ the set $  \mathcal{P}_{M}(b)  $ contains at least $ \varepsilon M! $ elements. Then:
$$
T \geq \frac{\varepsilon^2 b (M-1)  - 2 \varepsilon\sqrt{2 \varepsilon (M-1)}}{16 \, \sigma_M(u)} \, , \qquad [ =: T_M(b,\varepsilon) ]
$$
where $ \sigma_M(u) := \sqrt{\sum_{k = 1}^M ( u(k) - u(j_0))^2 } $ is assumed to be strictly positive.
\end{theorem}
The proof of this theorem is found Ref.~\onlinecite{FGGN08}. \\

If the energy gaps of $ u $ are of order one, the quantity $ \sigma_M(u) $ will be of order $ \sqrt{M} $. 
The above theorem, then implies that the quantum search algorithm is not faster than order $ \sqrt{M} $ -- the timescale of the Grover algorithm.\cite{Gr1,Gr2} This is a well-known fact which has been discussed early on in various special cases.\cite{DMV01,FGGN08}

\subsection{\label{sec:level2} A  gap estimate in the scambled ensemble}

Fahri, Goldstone, Gosset, Gutmann, and Shor\cite{FGGGS10} now combined the lower bound on the run time with the adiabatic theorem to obtain an upper bound on the gap $ \gamma_\pi(s) $ above the ground-state energy of the family of scambled Hamiltonians $ h_\pi(s) = h_D(s) + c(s) U_\pi $, or more precisely on
\begin{equation}
 \gamma_{\min,\pi}^\#:=  \min_{s \in [0,1] }  \min\left\{ \gamma_\pi(s)^3, \gamma_\pi(s)^2  \right\} \, . 
\end{equation}
Their argument proceed as follows. 
The adiabatic theorem (Theorem~\ref{thm:adabatic}) yields for all $ T > 0 $:
\begin{equation}\label{eq:adabcons}
\sqrt{1 - \left| \langle \psi_\pi(T) , e_{\pi(j_0)} \rangle \right|^2 } \leq \frac{ n_M}{T \gamma_{\min}^\# } \, . 
\end{equation}
where $ n_M := \max_{\pi\in \mathcal{S}_M} \left( 9 \max_{s \in [0,1]} \max\{\left\| h_\pi'(s) \right\|, \left\| h_\pi'(s) \right\|^2 \} + \max_{s \in [0,1]} \left\| h_\pi''(s) \right\| \right)  $. Consider now $ \varepsilon \in (0,1] $ and $ M \geq \max\{ 4, 128/\varepsilon \} $  such that  
\begin{equation}\label{eq:Test}
T_M(\tfrac{1}{2}, \varepsilon) \geq \frac{\varepsilon^2 M }{128 \sigma_M(u) }  > 0 \, . 
\end{equation} 
The adiabatic estimate~\eqref{eq:adabcons} with  $ T =  T_M(\tfrac{1}{2}, \varepsilon)/2 $ implies that for all permutations $ \pi \in \mathcal{S}_M $ for which $  \gamma_{\min,\pi}^\# \geq  \frac{2\sqrt{2}  \, n_M}{T_M(\frac{1}{2}, \varepsilon) } $, the search algorithm succeeds with probability $  \left| \langle \psi_\pi(T) , e_{\pi(j_0)} \rangle \right|^2 \geq \frac{1}{2} $. By Theorem~\ref{thm:QAC} 
this implies that the set of such permutation, 
\begin{equation}
\mathcal{G}_M(\varepsilon) := \left\{ \pi \in \mathcal{S}_M \, | \,  \gamma_{\min,\pi}^\# \geq \frac{2\sqrt{2}  \, n_M}{T_M(\frac{1}{2}, \varepsilon) } \right\} \, , 
\end{equation}
can only make up a fraction of at most $ \varepsilon $ of the total number $ M! $ of permutations. Otherwise one would have a contradiction to  Theorem~\ref{thm:QAC}. This is summarized in the following corollary taken from Ref.~\onlinecite{FGGGS10}.

\begin{corollary}[cf.~Ref.~\onlinecite{FGGGS10}]\label{cor:gap}
Assume that the family of scambled Hamiltonians $ h_\pi(s) = h_D(s) + c(s) U_\pi $ satisfies the assumptions {\bf a1-3} with $ h $ replaced by $ h_\pi $  for all $ \pi \in \mathcal{S}_M $ and all $ s \in [0,1] $. Then for all
$ \varepsilon \in (0,1] $ and $ M \geq \max\{ 4, 128/\varepsilon \} $:
\begin{align}\label{eq:gap}
\left|\mathcal{G}_M(\varepsilon) \right|  \leq \varepsilon M! \, .
\end{align}
\end{corollary}

\section{Application: QREM}\label{sec:QREM}

Among the physically relevant examples of unstructured energy landscapes are spin glasses. The simplest (mean-field version) is the random energy model (REM) by Derrida in which one considers the configuration space 
$ \mathcal{Q}_N = \{0, 1\}^N$ of  $ N $ Ising spins.\cite{Der80,Bov06} To each of these $ M = 2^N $ spin configurations, one assigns a random energy
\begin{equation}\label{eq:REM}
u(\sigma) = \sqrt{N} \; g(\sigma)  \, , \qquad  \sigma \in \mathcal{Q}_N  \, ,
\end{equation}
where
$\{ g(\sigma) \}_{\sigma \in Q_N}  $ are independent and identically standard normally distributed random variables. 
The scaling factor in \eqref{eq:REM} ensures that the values of $ u $ are found on in the range 
$ - \frac{N}{\kappa_c} \lesssim u(\sigma)\lesssim  \frac{N}{\kappa_c} $  with 
$$
 \kappa_c = \frac{1}{\sqrt{2 \ln 2}} \, . 
 $$
 More precisely, for $ x > - \ln N / \ln 2 $ let $ v_N(x) \in \mathbb{R} $ be the unique solution of 
 $
 \int_{v_N(x)}^\infty e^{-t^2/2} \frac{dt}{\sqrt{2\pi}} \ = \  2^{-N} e^{-x} $. Then 
 $$
 \sqrt{N} \,  v_N(x) \ = \ \frac{N}{\kappa_c}   + \kappa_c \, x  - \frac{\kappa_c}{2 } \ln \left( 4\pi \ln 2^N \right) + o(1) \, , \qquad N \to \infty \, ,
 $$
and the extremal value statistics of the REM reads:
 \begin{proposition}[cf.~Ref.~\onlinecite{LLR,Bov06}]\label{prop:exval}
 The distribution of the minimum of the REM is asymptotically as $ N \to \infty $ given by 
 \begin{equation}\label{eq:exvaluestat} 
 \mathbb{P}\left( \min u \geq -  \sqrt{N} \, v_N(x) \right) \ = \ \left(1-2^{-N} e^{-x}\right)^{2^N} \rightarrow e^{-e^{-x}} \, . 
 \end{equation}
 Moreover, the process 
 $
 -(v_N)^{-1}(-g(\sigma)) \ = \ \frac{u(\sigma)}{\kappa_c} + \frac{N}{\kappa_c^2} -  \frac{1}{2 } \ln \left( 4\pi \ln 2^N \right) + o(1) $, $ \sigma \in  \mathcal{Q}_N $, 
 converges in distribution to a Poisson process with intensity measure $ e^{\tau} d\tau $. 
 \end{proposition}
Since the extremal small values of the REM converge to a Poisson process,  the ground-state of the REM is typically separated by order one from the first excited state. \\

One may render $ \mathcal{Q}_N $ a graph by declaring vertices $ \sigma, \sigma' \in \mathcal{Q}_N  $ as neighbours, i.e.\ $ \sigma' \sim \sigma $, if  they differ by one spin flip. The graph Laplacian on this so-called Hamming-cube is then given by 
$$
\left(\Delta \psi\right)(\sigma) = \sum_{\sigma' \sim \sigma}  \psi( \sigma' ) - N \psi( \sigma) \, , \quad \psi \in \ell^2( \mathcal{Q}_N) \cong \mathbb{C}^{2^N} \, . 
$$
By identifying the canonical basis in $ \mathbb{C}^{2^N} $ with the joint eigenbasis of the third-components $ \sigma_j^z $, $ j= 1, \dots , N $, of the spin-operators of $ N $ spin-$1/2 $ particles, the Laplacian may be interpreted as a transversal constant magnetic field on those spins, $ -\Delta = N - \sum_{j=1}^N \sigma_j^x $. Adding the REM energies in form of a diagonal matrix $ U $ gives rise to the quantum random energy model (QREM):
\begin{equation}
 \mathcal{H}(\kappa) = - \Delta + \kappa\,  U     \, , \quad \kappa > 0 \, .
\end{equation}
Among the interesting properties of this model is a first-order phase transition of the 
ground-state  of $  \mathcal{H}(\kappa) $ at $ \kappa = \kappa_c $. Numerical findings of J\"org, Krzakala, Kurchan and Maggs\cite{JKKM08} suggest that:
\begin{description}
\item[Case $ \kappa < \kappa_c $] the ground-state is delocalised  with energy $E_0(\kappa) = -\kappa^2  + o(1) $ whose fluctuations are suppressed exponentially in $ N $.
\item[Case $ \kappa > \kappa_c $]  the ground-state is localised approximately in the eigenvector corresponding to the unique minimum of $ u $
with energy $ E_0(\kappa)  = N + \kappa \, \min u  +  \mathcal{O}(1) = N ( 1- \frac{\kappa}{\kappa_c} ) +   \mathcal{O}(\ln N) $,
\item[Case $ \kappa = \kappa_c $] The energy gap $ \Gamma(\kappa) = E_1(\kappa) - E_0(\kappa)  $  above the unique ground-state 
	closes exponentially in $ N $.
\end{description} 
In this context, it is useful to recall that the spectrum of the Laplacian $ \mathcal{H}(0) $ can be easily computed (as a sum of $ N $ commuting operators). It coincides with the even integers $ \{ 0, 2, \dots , 2N \} $ and the unique ground-state is the maximally delocalised state $ \phi(0) =  \frac{1}{\sqrt{2^N}}  (1, \dots , 1 )^T \in \mathbb{C}^{2^N}$. \\

The full justification of the above sketched low-energy properties of the QREM will be the topic of another paper.\cite{sw} 
Our main aim here is to point out that the conjectured vanishing of the gap $ \Gamma(\kappa) $ at some $ \kappa > 0 $ is a straightforward corollary of the general considerations in the first section.
\begin{theorem}\label{thm:main}
There is a numerical constant $ C < \infty $ such that the energy gap above the unique ground-state of the QREM satisfies:
\begin{equation}\label{eq:bound}
\lim_{N\to \infty} \mathbb{P}\left( \min_{\kappa \in (0,N^3)} (E_1(\kappa) - E_0(\kappa) ) \leq C \,  N^{5} \, 2^{-\frac{N}{6}} \right) = 1 \, . 
\end{equation}
\end{theorem}
Before giving the short proof let us add a few comments:
\begin{enumerate}
\item The arguments only yield the existence of some value of the coupling $ \kappa \in ( 0, N^3) $ at which the gap closes exponentially and do not determine the conjectured value $ \kappa = \kappa_c $. In particular, the value of the critical coupling could still be dependent on  $ N $ and the realisation of the REM.  This can only be excluded in a more detailed analysis.\cite{sw}

Nevertheless, the established bound~\eqref{eq:bound} already distinguishes the low-energy properties of the QREM from those of the REM, where the energy gap above the ground state is typically order one, cf.~Proposition~\ref{prop:exval}. 

As will be seen in the subsequent proof, the upper bound $ \kappa < N^3 $ on the interval in which the phase transition can occur is far from being optimized. In order to keep the paper simple, we however refrain from optimising this value. 

\item The fact that first-order phase transitions of the ground-state are the stumbling block to speeding up polynomially the search in various problems in spin-glass theory is well-known - the REM landscape is just one example. Other interesting examples are random optimisation problems from the SAT class (for instances having a unique satisfying assignment), see Refs.~\onlinecite{PO1,Alt09,AC09,J+10} and the recent review Ref.~\onlinecite{Bapst13} and references therein. 
\end{enumerate}
\begin{proof}[Proof of Theorem~\ref{thm:main}]
We aim to apply Corollary~\ref{cor:gap} with $ M = 2^N $ and
$$
h_\pi(s) = - (1-s) \Delta + s\,  U_\pi \, , \quad s \in [0,1] \, .
$$
To do so, we note that Assumption~{\bf a1} as well as {\bf a2} are evidently satisfied. It remains to check {\bf a3}. Since $ h_\pi(s) $ generates for each $ s \in [0,1) $ and $ \pi $ a positivity improving semigroup, the ground-state of $ h_\pi(s) $ is unique by the Perron-Frobenius theorem. In case $ h_\pi(1) = U_\pi $ the almost-sure uniqueness of the ground-state follows from the almost-sure non-degeneracy of the $ 2^N $ Gaussian random variables. Moreover, we may estimate
\begin{align*}
& \sigma_{M}(u) \leq \sqrt{M} \, 2 \| u \|_\infty \, , \\
&  \left\| h_\pi'(s) \right\| \leq \| \Delta \| + \| U \| \leq 2 N +  \| u \|_\infty \, ,
\end{align*}
and $ h''(s) = 0 $. For all realisations of the REM aside from a fraction whose probability vanishes exponentially as $ N\to \infty $, we also have
$
\| u \|_\infty = \max_\sigma |u(\sigma) | \leq 2N/ \kappa_c
$. 
This follows from the extremal value statistics~\eqref{eq:exvaluestat} of the REM. Thus $ n_M \leq C N^2  $ with some numerical constant $ C < \infty $ and we may conclude from~\eqref{eq:gap} with $ \varepsilon = N^{-1} $ that for all $ N $ large enough 
$$
\frac{1}{M!} \sum_{\pi \in \mathcal{S}_M} 1\left[\min_{s \in [0,1] }  \min\left\{ \gamma_\pi(s)^3, \gamma_\pi(s)^2  \right\}  \leq  \,   C \frac{N^{9/2}}{\sqrt{M}} \right] \geq 1 - N^{-1} \, ,
$$
where $ C < \infty $ is again a numerical constant. The right side in the indicator function $ 1[ \dots ] $ is smaller than one for $ N  $ large enough, such that that 
$
\frac{1}{M!} \sum_{\pi \in \mathcal{S}_M}  1\left[\min_{s \in [0,1] } \gamma_\pi(s) \leq C N^{3/2} / M^{1/6} \right] \geq 1 - N^{-1} $. 
We now use the permutation invariance of the distribution  of the REM to conclude for all $ N $ large enough
\begin{equation}
\mathbb{P}\left( \min_{s \in [0,1] } \gamma(s) \leq C N^{3/2} / M^{1/6} \right) = \frac{1}{M!} \sum_{\pi \in \mathcal{S}_M}  \mathbb{P}\left( \min_{s \in [0,1] } \gamma_\pi(s) \leq C \frac{N^{3/2}}{M^{1/6}} \right) \geq 1 - 2 N^{-1} \, . 
\end{equation}
(The factor of two accounts for disregarding all realisations for which $ \| u \|_\infty > 2 N/\kappa_c $ which occur even with exponentially small probability.)
In order to relate the QREM to $ h(s) $, we write
$$
\mathcal{H}(\kappa) = (1+\kappa) \, h\left(\frac{\kappa}{1+\kappa}\right) \, 
$$
such that $ \Gamma(\kappa) = (1+\kappa) \, \gamma\left(\frac{\kappa}{1+\kappa}\right) $ and hence $ \lim_{N\to \infty} \mathbb{P}\left( \min_{\kappa \geq  0} \Gamma(\kappa) /(1+\kappa) \leq C \,  N^{3/2} \, 2^{-\frac{N}{6}} \right) = 1 $.

The fact that the minimum value of the ratio $ \Gamma(\gamma)/ (1+\kappa) $ is attained at $ \kappa > 0 $ is elementary. That it is attained with asymptotically full probability at some $ \kappa < N^3 $ is seen using the variational principle. The latter yields the following elementary bounds $ E_1(\kappa) \geq \kappa u_1 $ and $ E_0(\kappa) \leq N + \kappa u_0 $, where $ \min u =:u_0 < u_1 $ denote the minimum and second smallest value of the REM energies. Since the difference $ u_1-u_0 $ is bounded from below by $ N^{-1} $ with asymptotically full probability (cf.~Proposition~\ref{prop:exval}), the resulting lower bound $ \Gamma(\kappa)/ (1+\kappa) \geq \kappa (u_1-u_0) / (1+\kappa)- N / (1+\kappa)$ is bounded from below by a positive constant times $  N^{-1} $ for all $ \kappa \geq N^3  $ with asymptotically full probability. This completes the proof.
\end{proof}

\begin{acknowledgments}
This work was accomplished as a summer project of J.A. at the Zentrum
Mathematik at TUM financed by MMUF, the Mellon Mays Undergraduate
Fellowship.  J.A. would like to thank TUM for their hospitality; in
particular, Frau B\"acker for her secretarial support in organizing  a
work environment.
S.W. is supported by the DFG (WA 1699/2-1).
\end{acknowledgments}

\end{document}